\documentclass[12pt]{article}
\usepackage{graphicx} % Required for inserting images

\usepackage{etoolbox}
\AtBeginEnvironment{array}{\setlength{\arraycolsep}{0.278em}}

\usepackage{fullpage}
\usepackage[all]{xy}
\usepackage[usenames,dvipsnames]{xcolor}
\usepackage{longtable}
\usepackage{enumerate}
\usepackage{hyperref}
\usepackage{underscore}
\usepackage{cite}
\usepackage{bm}
\usepackage{mathtools}

\usepackage{amsmath} 
\usepackage{verbatim}
\usepackage{textcomp}
\usepackage{amsfonts}
\usepackage{amssymb}
\usepackage{amsthm}
\usepackage{fancyhdr}
\usepackage{graphicx}
\usepackage[english]{babel}
\usepackage{wrapfig}
 \usepackage{caption} \usepackage{subcaption} 
\usepackage[margin=1in]{geometry}
\usepackage{tikz} 
\usepackage[normalem]{ulem}
\usepackage{listings}
\usetikzlibrary{arrows}

\lstdefinestyle{Sage}{%
    language = Python,
    commentstyle = \color{OliveGreen},
    keywordstyle = \color{blue}    
}

\newcommand{\F}{\mathbb{F}}

\newcommand{\bpf}{\begin{proof}}
\newcommand{\epf}{\end{proof}}

\newcommand{\wt}{\mathrm{wt}}

\definecolor{darkviolet}{rgb}{0.58, 0.0, 0.83}
\newcommand{\rmv}[1]{}

\def\F{\mathbb F}

\def\supp{\mathrm{supp}\,}

\newcommand{\FF}{\mathbb{F}}
\newcommand{\NN}{\mathbb{N}}

\newtheorem{theorem}{Theorem}[section] 
\newtheorem{proposition}[theorem]{Proposition} 

\newtheorem{corollary}[theorem]{Corollary} 

\theoremstyle{definition}

\newtheorem{example}[theorem]{Example} 

\newtheorem{definition}[theorem]{Definition} 
\newtheorem{remark}[theorem]{Remark}
 
\newtheorem{obs}[theorem]{Observation} 

\newcommand\blfootnote[1]{%
  \begingroup
  \renewcommand\thefootnote{}\footnote{#1}  \addtocounter{footnote}{-1}%
  \endgroup
}

\title{Influences of some families of error-correcting codes}

\author{Hailey Egan, Jason~T.~LeGrow\thanks{Jason T.~LeGrow is supported in part by the Commonwealth of Virginia's Commonwealth Cyber Initiative (CCI), an investment in the advancement of cyber R\&D, innovation, and workforce development. For more information about CCI, visit \url{www.cyberinitiative.org.}}, Gretchen~L.~Matthews\thanks{G. L. Matthews is supported in part by NSF DMS-1855136 and the Commonwealth Cyber Initiative.}, Jeff Suliga\thanks{Support for this project was provided in part by the Commonwealth Cyber Initiative, the Virginia Tech Department of Mathematics, and the MAA Tensor Foundation.}
\blfootnote{\noindent Author list in alphabetical order; see \url{https://ams.org/profession/leaders/ CultureStatement04.pdf.}}}

\date{ }

\begin{document}

\maketitle

\begin{abstract} %150 word max
 The ability of a linear error-correcting code to recover erasures is connected to influences of particular monotone Boolean functions. These functions provide insight into the role that particular coordinates play in a code’s erasure repair capability. In this paper, we consider directly the influences of coordinates of a code. We describe a family of codes, called codes with minimum disjoint support, for which all influences may be determined. As a consequence, we find influences of repetition codes and certain distinct weight codes. Computing influences is typically circumvented by appealing to the transitivity of the automorphism group of the code. Some of the codes considered here fail to meet the transitivity conditions required for these standard approaches, yet we can compute them directly.
\end{abstract}

\section{Introduction}

Capacity-achieving codes are considered the holy grail by information theorists, since They provide reliable communication at the most efficient rates. Their existence---proven in Shannon's theorem \cite{shannon}---launched an entire discipline as researchers search for explicit constructions with efficient encoding and decoding algorithms. Recently, algebraic tools have been used by Kudekar, Kumar, Mondelli, Pfister, Şaşoǧlu, and Urbanke to demonstrate the capacity-achieving nature of Reed-Muller codes, settling an old problem in coding theory \cite{Kudekar_17}. The authors connect this coding theory question to influences of variables in Boolean functions, a notion introduced by Ben-Or and Linial \cite{Ben-Or}  for collective coin flipping which has since been used in a variety of contexts; see, for instance, the survey \cite{survey}. 

In this paper, we explore influences as they relate to coordinates of error-correcting codes and the ability of a codeword symbol to be recovered from a received word in which erasures have occurred.  Following the breakthrough work \cite{Kudekar_17}
which  relies on the fact that the automorphism group of a Reed-Muller is doubly transitive, Kumar, Calderbank, and Pfister 
\cite{Kumar_16} observed that a weaker condition on the automorphism group that can be used to prove that a code is capacity achieving. They study the behavior of influences of monotone Boolean functions $f$ that encapsulate which errors are correctable when the orbits of the automorhism group of the code are large. This approach can be thought of as a proxy for working directly with the influences themselves, relying instead on the lower bound on the largest influence provided by the celebrated Kahn-Kalai-Linial Theorem \cite{KKL_88} and the high degree of symmetry of the code. 

A {\bf Boolean function} $f\colon \F_2^M \rightarrow \F_2$ is said to be {\bf monotone} if $f(x) \leq f(y)$  for all $x, y \in \F_2^M$ with $x_i \leq y_i$ for all $i=1, \dots, M$. Here, $\F_2=\{ 0 , 1 \}$ denotes the field with two elements, and we consider the ordering in which $0 \leq 1$. For a real number $0<p<1$, the {\bf influence of coordinate $j$ on $f$}, or simply the influence of $j$, is 
$$
I_j^{(p)}(f) := \sum_{\substack{x \in \F_2^n\\ f(x) \neq f(x+e_j)}} p^{\wt(x)}(1-p)^{M-\wt(x)}
$$
where the {\bf weight} $\wt(x)$  of a vector $x$ is its number of nonzero coordinates and $e_j=(0, \dots, 0, 1, 0, \dots, 0)$ is the standard basis vector whose only nonzero coordinate is in the $j^\text{th}$ position. Notice that $I_j^{(p)}(f)$ may be thought of as the probability of the subset
$$\Omega_j(f):=\left\{ x \in \F_2^M\colon f(x) \neq f(x+e_j) \right\}$$
under
the Bernoulli-$p$ product measure. The {\bf total influence} $I^{(p)}(f):=\sum_{j \in [M]} I_j^{(p)}(f)$ appears in the bound
$$
p_{1-\varepsilon} - p_{\varepsilon} \leq \frac{2\mu_p (\Omega_f) (1-\mu_p (\Omega_f))}{I^{(p)}(f)} \log \frac{1-\varepsilon}{\varepsilon}
$$
on the transition width of the EXtrinsic Information Transfer (EXIT) function \cite{Kudekar_17}. There, code symmetry and a bound on the value $I^{(p)}(f)$ are used to show that the EXIT function exhibits a sharp transition, implying that the code is capacity achieving. As a step toward applying these ideas to more general families of codes, we consider the influences themselves. In particular, we pursue directly the influences of codes with a lower degree of symmetry, noting that theoretical results on influences tend to be elusive. We specify code families for which influences can be expressed for an arbitrary member and describe them completely. While they do not reveal any new families of capacity achieving codes, we hope they will provide insight into behaviors of influences of coordinates of codewords. 

To our knowledge, we are the first group to explicitly compute influences for families of codes in this way. Our study of influences is motivated by their appearance in the proof that some code families, including the Reed-Muller codes, achieve capacity. However, the code families here do not provide new instances of capacity-achieving codes due to their low dimensions. In particular, for each family considered here, as the code length grows, the code rate goes to zero. Hence, it remains an open problem to find new families of capacity-achieving codes by computing influences directly.

This paper is organized as follows. Section 
\ref{prelim_section} reviews terminology needed from coding theory, provides a motivating example, and describes the connection between influences of codes and monotone Boolean functions.  In Section~\ref{min_section}, we consider codes with controlled support, defining so-called {\bf minimal disjoint support codes}, and determine their influences. Finally, we give a summary and discussion of open problems in Section~\ref{conclusion_section}.

\section{Preliminaries} \label{prelim_section}

An $[n,k,d]$ code $C$ over  $\F_2$ is a $k$-dimensional $\F_2$-subspace of the vector space $\F_2^n$ in which any two distinct elements $c, c' \in C$ differ in at least $d$ coordinates. The code $C$ is said to have {\bf length} $n$, {\bf dimension} $k$, and {\bf minimum distance} $d$. We sometimes omit \(d\) from the notation and call \(C\) an \([n,k]\) code. Elements of $C$ are called {\bf codewords}. When using an \([n,k]\) code $C$, received words are elements of $\left(\F_2 \sqcup \{ * \} \right)^n$ where $*$ denotes an erasure. 
Associated with a received word $w \in \left(\F_2 \sqcup \{ * \} \right)^n$ is an {\bf erasure pattern} $e \in \F_2^n$ whose entries are given by
$$e_i:= \begin{cases} 1 & \textnormal{if } w_i=* \\ 0 & \textnormal{otherwise}. \end{cases}$$

The code $C$ may be described by a {\bf generator matrix} $G$, meaning a matrix whose rows span $C$.  We may consider $xG$ as the encoding of a message vector $x \in \F_2^k$ and typically take $G \in \F_2^{k \times n}$ so that
$$C = \left\{ x G \colon x \in \F_2^k \right\},$$
where $\F_2^{m \times n}$ denotes the set of $m \times n$ matrices with entries in $\F_2$. Often \(G\) will appear in {\bf systematic form}, meaning that the first \(k\) columns of \(G\) are the \(k \times k\) identity matrix. The code $C$ can be also described by a {\bf parity-check matrix} $H \in \F_2^{(n-k) \times n}$, which is a matrix that has $C$ as its nullspace; that is, 
$$
C=\left\{ x  \in \F_2^n \colon Hx^T = 0 \right\}.
$$
When \(C\) has generator matrix \(G = [I | P]\) in systematic form, a parity-check matrix for \(C\) is \(H = [P | I]\) (in characteristic 2).

The {\bf product} of an $[n_1,k_1]$ code $C_1$ and an $[n_2,k_2]$ code $C_2$ is defined to be the code \(C_1 \otimes C_2\) with generator matrix  
\[ G_1 \otimes G_2 = \left(\begin{array}{cccc} a_{11}G_2 & a_{12}G_2 & \cdots & a_{1n_1}G_2 \\ a_{21}G_2 & a_{22}G_2 & \cdots & a_{2n_1}G_2 \\
\vdots & \vdots &  & \vdots \\ a_{k_11}G_2 & a_{k_12}G_2 & \cdots & a_{k_1n_1}G_2 \\ \end{array}\right)
\in \F_{2}^{k_1k_2\times n_1n_2}
\] where $G_1 =
\left(\begin{array}{cccc} a_{11} & a_{12} & \cdots & a_{1n_1} \\ a_{21} & a_{22} & \cdots & a_{2n_1} \\
\vdots & \vdots &  & \vdots \\ a_{k_11} & a_{k_12} & \cdots & a_{k_1n_1} \\ \end{array}\right) \in \F_{2}^{k_1\times n_1}$ is a generator matrix of $C_1$ and  $G_2  \in \F_{2}^{k_2\times n_2}$ is a generator matrix of the code $C_2$. Equivalently, \(C_1 \otimes C_2\) is simply the linear span of \( \{c_1 \otimes c_2 \,:\, c_1 \in C_1, c_2 \in C_2\}\).

Given a positive integer $n$, set $[n]:=\left\{ 1, \dots, n \right\}$.
We refer to elements of $\F_2^n$ as {\bf words} or {\bf vectors}. 
 We write $A \sqcup B$ to mean the disjoint union of sets $A$ and $B$ (typically, subsets of $[n]$ or some $\F_2^{n'}$ with $n' \leq n$).

\section{Influences of coordinates}
In this section, we use a scenario to describe the connection between influences of coordinates of codewords to influences of related monotone Boolean functions and erasure recovery. We then formalize the notion and provide some examples. 

We begin with some fundamental definitions. Let $C$ be an \([n,k]\) code. The {\bf support} of $x \in \F_2^n$ is $$\supp(x) := \left\{ i \in [n]\colon x_i \neq 0 \right\}.$$ Notice that the weight of $x \in \F_2^n$ satisfies $$\wt(x) =  | \left\{ i \in [n]\colon x_i \neq 0 \right\} | = | \supp (x) |.$$
\begin{definition}
Given $x,x' \in \F_2^n$, we say that $x$ {\bf covers} $x'$ if and only if $\supp(x) \supseteq \supp(x')$. In this case, we write $x \succeq x'$. We may also write $x' \preceq x$ 
and say that $x'$ is {\bf covered by} $x$.
\end{definition}

We are now prepared for a motivating example. Consider the code \(C\) with systematic generator matrix
\[
G = \left[ \begin{array}{rrrr} 1 & 0 & 0 & 0 \\ 0 & 1 & 0 & 0 \\ 0 & 0 & 1 & 0 \\ 0 & 0 & 0 & 1 \end{array} \right| \left.
\begin{array}{rrrrr}
1 & 0 & 1 & 0 & 1 \\
1 & 0 & 1 & 1 & 0 \\
0 & 1 & 1 & 0 & 0 \\ 
1 & 1 & 1 & 1 & 1 \\
\end{array}
\right] \in \F_2^{4 \times 9}
\]

We would like to know what kind of erasures cannot be recovered by \(C\). At the highest level, we are unable to recover the underlying codeword \(c\) from a partially erased codeword \(w \in \left(\FF_2 \sqcup \{*\}\right)^9\) if there exist two or more plausible codewords for \(w\); that is, two distinct codewords \(x, y\) which agree with \(w\) on the indices which have not been erased. For instance, consider the received word 
\[
w = \begin{array}{ccccccccc} *& *& 1& *& 0& 0& 1& *& * \end{array}
\]
which we may think of as a partially-erased codeword. There are (at least) two plausible codewords:
\begin{align*}
x &= \begin{array}{ccccccccc} 1& 0& 1& 1& 0& 0& 1& 1& 0\end{array}\\
y &= \begin{array}{ccccccccc} 0& 1& 1& 1& 0& 0& 1& 0& 1\end{array}
\end{align*}

This is equivalent to saying that \(x-y\) is zero at all indices which are not erased in \(w\). Since \(C\) is a linear code and $x \neq y$, \(x-y\) is a (nonzero) codeword which is identically zero on the coordinates which are not erased. Consequently, if we had received a partially-erased codeword \(w'\) with the same erasure pattern as \(w\), and in which each non-erased coordinate was \(0\), we could not determine whether the sent codeword was \(\vec{0}\) or \(x-y\). 

This ambiguity arises when the erasure pattern erases all non-zero coordinates of any non-zero codeword. This naturally brings us to an equivalent, more convenient way to determine whether an error is recoverable: 

\begin{obs}
A received word is recoverable from an erasure pattern \(e \in \F_2^n\)  if and only if \(e\) does not cover any non-zero codeword of \(C\).
\end{obs}

In the case of this example, we had
\[
e = \begin{array}{ccccccccc} 1& 1& 0& 1& 0& 0& 0& 1& 1 \end{array}
\]
which covers the codeword
\[
x-y =  \begin{array}{ccccccccc} 1& 1& 0& 0& 0& 0& 0& 1& 1 \end{array}
\]
and so, given the partially-erased codeword
\[
w' = \begin{array}{ccccccccc} *& *& 0& *& 0& 0& 0& *& * \end{array}
\]
we cannot recover, because we cannot determine whether the correct codeword is \(x-y\) or \(\vec{0}\).

%Recall that all codes considered in this paper are linear and thus contain the all-zero codeword $(0, 0, \dots, 0)$.
\iffalse
To consider the impact of erasures on recovery using a such a code $C$, we can focus on the all-zero codeword. Observe that it can be recovered from a received word if and only if the erasure pattern does not cover the nonzero coordinates of a codeword. For instance, an erasure pattern of 
\begin{align*}
\begin{array}{ccccccccccc}
e&=&*&*&0&*&*&0&0&*&*
\end{array}
\end{align*}
makes it impossible to ascertain if 
\begin{align*}
\begin{array}{ccccccccccc}
c&=&1&1&0&1&1&0&0&0&1  
\end{array}
\end{align*}
or 
\begin{align*}
\begin{array}{ccccccccccc}
c'&=&0&0&0&0&0&0&0&0&0  
\end{array}
\end{align*}
was sent. Here, the impact of $e$ is that symbols in all but the third, sixth, and seventh coordinates are erased. To describe the interplay between erasure patterns and recovery of particular coordinates, we introduce some useful tools in the next subsection. 
\fi

Beyond full codeword recovery, we may instead be interested in recovering only specific coordinates of a partially-erased codeword---this is a slightly more delicate problem than that described above. We now introduce some useful tools to handle this situation.

The set of vectors covering a codeword with $i$ in its support is
\[
\Omega_i(C) := \{ {x} \in \FF_2^n \,\colon\, {x} \succeq {c} \mbox{ for some } c \in C \mbox{ with } c_i \neq 0 \};
\]
we write $\Omega_i$ if the code is clear from the context. It may be helpful to think about the set of codewords with $i$ in their supports:
$$S_i:= \left\{ c \in C \colon c_i \neq 0 \right\} = \left\{ c \in C \colon i \in \supp (c) \right\} \subseteq \F_2^n.$$
With this notation in mind, we have
\[
\Omega_i = \{ {x} \in \FF_2^n \,:\, {x} \succeq {c} \mbox{ for some } c \in S_i  \}
\] 
and $$| \Omega_i | \geq | S_i |.$$ 
The codeword symbol in position $i \in [n]$ can be recovered from a received word impacted by erasure pattern $e \in \F_2^n$ if and only if $$e \nsucceq c \ \forall c \in S_{i}.$$ Hence, 
the set of erasure patterns that prevent recovery of $i \in [n]$ is precisely $\Omega_i$.

\begin{definition} Given a code $C$ of length $n$ and $j \in [n] \setminus \{ i \}$, the {\bf \(j^\text{th}\) boundary of \(\Omega_i\)} is
\[
\partial_j \Omega_i := \{{x} \in \FF_2^n \,:\, {x} \in \Omega_i \wedge {x} + {e}_j \not\in \Omega_i \} \cup \{ x \in \F_2^n\colon  {x} \not\in \Omega_i \wedge {x} + {e}_j \in \Omega_i\}.
\]
The {\bf \(j^\text{th}\) boundary} of $C$ is 
\[
B_j: = \bigcup_{i \in [n] \setminus \{  j \} } \partial_j\Omega_i;
\]
\end{definition}

We may think of the \(j^\text{th}\) boundary of \(\Omega_i\),  \(\partial_j\Omega_i\), as the set of all vectors \({x} \in \F_2^n\) such that changing the \(j^\text{th}\) coordinate of \({x}\) ``toggles'' whether \({x}\) is in \(\Omega_i\). 

To determine how influential $j \in [n]$ is on recovery of $i \in [n] \setminus \{ j \}$, observe that the set $\partial_j\Omega_i$ consists of erasure patterns where changing the $j^\text{th}$ coordinate either ``moves'' a vector
\begin{itemize}
\item 
 from $\Omega_i$ to $\F_2^n\setminus \Omega_i$, meaning changing the $j^\text{th}$ coordinate converts an erasure pattern that was not recoverable into one that is, or
\item  from $\F_2^n\setminus \Omega_i$
to $\Omega_i$, meaning changing the $j^\text{th}$ coordinate converts an erasure pattern that was recoverable into one that is not.
\end{itemize}
Loosely speaking, \(B_j\) is the set of  vectors for which flipping their \(j^\text{th}\) entry results in ``crossing'' the \(j^\text{th}\) boundary of some \(\Omega_i\). 
The influence of $j \in [n]$ on recovery is then found by considering how influential $j$ is on all $\Omega_i$ where $i \in [n] \setminus \{j\}$ giving rise the the definition below. 

\begin{definition}
Given $0 <p<1$, the {\bf influence of the $j^\text{th}$ coordinate of a code} $C$ is defined by
\begin{equation} \label{inf_def}
I_j^{(p)}(C) = \sum_{{x} \in B_j} p^{\wt({x}) - 1}(1-p)^{n - \wt({x})}
\end{equation} and the {\bf total influence of (the coordinates of) a code $C$} is $$I^{(p)}(C)=\sum_{j \in [n]} I_j^{(p)}(C).$$
\end{definition}
We often write $I^{(p)}_j$ and $I^{(p)}$ when the code \(C\) is clear from the context.

\begin{remark}
Note that $I_j^{(p)}(C)$
intuitively measures the probability that changing the  \(j^\text{th}\) coordinate \(x_j\) of $x\in \F_2^n$ causes \({x}\) to cross the \(j^\text{th}\) boundary of some $\Omega_i$. Observe that if \(x \in \partial_j\Omega_i\) then certainly \(x_i = 1\), since if \(x_i = 0\) then \( (x+e_j)_i = 0\) so that neither \(x\) nor \(x+e_j\) can be in \(\Omega_i\).
Then 
$$
I_j^{(p)}(C)  \leq \sum_{{x} \in \F_2^n} p^{\wt({x}) - 1}(1-p)^{n - \wt({x})} = \frac{1}{p} \sum_{\ell=0}^n \binom{n}{\ell} p^{\ell}(1-p)^{n-\ell} = \frac{1}{p}.
$$
\end{remark}

We now consider a toy example to illustrate these concepts.
 
\begin{example} \label{toy_ex}
Consider the $[5,2,3]$ code $C$ with generator matrix 
$$
\left[ 
\begin{array}{ccccc}
1 & 1 & 1 & 0 & 0 \\
0 & 0 & 1 & 1 & 1
\end{array}
\right]. 
$$
Observe that $$C=\left\{ (1,1,1,0,0),  (0,0,1,1,1), (1,1,0,1,1), (0,0,0,0,0) \right\}.$$ 
Considering the sets of codewords with particular indices in their supports, we see that $$S_1=S_2=\left\{ (1,1,1,0,0),(1,1,0,1,1) \right\},$$
$$S_3=\left\{ (1,1,1,0,0),(0,0,1,1,1) \right\},$$ and 
$$S_4=S_5=\left\{ (0,0,1,1,1),(1,1,0,1,1)   \right\}.$$
It follows that the sets of vectors covering those in the sets above are
$$\Omega_1=\Omega_2=\left\{ (1,1,1,0,0),(1,1,1,1,0),(1,1,1,0,1),(1,1,1,1,1),(1,1,0,1,1) \right\},$$
$$\Omega_3=\left\{ \begin{array}{l}(1,1,1,0,0),(1,1,1,1,0),(1,1,1,0,1),(1,1,1,1,1), \\ (0,0,1,1,1),(0,1,1,1,1),(1,0,1,1,1) \end{array} \right\},$$ and 
$$\Omega_4=\Omega_5=\left\{ (0,0,1,1,1), (1,0,1,1,1), (0,1,1,1,1), (1,1,1,1,1),(1,1,0,1,1)   \right\}.$$
To find the $j^{th}$ boundaries when $j=1$, we first note that 
$$\partial_1 \Omega_2
=\left\{ \begin{array}{l} (1,1,1,0,0), (1,1,1,1,0),(1,1,1,0,1),(1,1,1,1,1),(1,1,0,1,1),\\ (0,1,1,0,0),(0,1,1,1,0),(0,1,1,0,1),(0,1,1,1,1),(0,1,0,1,1) \end{array} \right\},$$
$$ \partial_1 \Omega_3
=\left\{ \begin{array}{l}(1,1,1,0,0), (1,1,1,1,0),(1,1,1,0,1),
(0,0,1,1,1),\\(0,1,1,0,0),(0,1,1,1,0),(0,1,1,0,1) \end{array} \right\},$$ and 
$$ \partial_1 \Omega_4
=\partial_1 \Omega_5=\left\{ (1,1,0,1,1),(0,1,0, 1,1)\right\}.$$As a result, 
$$
B_1=\left\{ 
\begin{array}{l}
(0,1,1,0,0),
(0,1,0,1,1),
(0,1,1,0,1),
(0,1,1,1,0),
(1,1,1,0,0), \\
(0,1,1,1,1),
(1,1,0,1,1),
(1,1,1,0,1),
(1,1,1,1,0),
(1,1,1,1,1)
\end{array}
\right\}.$$
Similarly, 
$$
B_2=\left\{ 
\begin{array}{l}
(0,1,1,0,0),
(0,1,0,1,1),
(0,1,1,0,1),
(0,1,1,1,0),
(1,1,1,0,0), \\
(0,1,1,0,1),
(0,1,0,1,0),
(0,1,1,0,0),
(0,1,1,1,1),
(1,1,1,0,1)
\end{array}
\right\}.$$
As a result, 
$$I_1^{(p)}=I_2^{(p)}=p(1-p)^3+4p^2(1-p)^2+4p^3(1-p)+p^4.$$
In addition, 
$$ \partial_3 \Omega_1
=\partial_3 \Omega_2=\left\{ \begin{array}{l}
(1,1,0,0,0),(1,1,0,0,1),(1,1,0,1,0),\\ (1,1,1,0,0),(1,1,1,0,1),(1,1,1,1,0) \end{array}
\right\}$$
and 
$$ \partial_3 \Omega_4
=\partial_3 \Omega_5=\left\{ \begin{array}{l}
(0,0,0,1,1),(0,0,1,1,1),(0,1,0,1,1), \\(1,0,0,1,1), (0,1,1,1,1),(1,0,1,1,1) \end{array}
\right\}.$$
Thus,
$$ B_3=\left\{ \begin{array}{l}
(0,0,0,1,1), (1,1,0,0,0),(0,0,1,1,1),(0,1,0,1,1),(1,0,0,1,1),(1,1,0,0,1),\\(1,1,0,1,0), (1,1,1,0,0),(0,1,1,1,1),(1,0,1,1,1),(1,1,1,0,1),(1,1,1,1,0) \end{array}
\right\}$$
As a result, 
$$
I_3^{(p)}=2p(1-p)^3+6p^2(1-p)^2+4p^3(1-p).
$$
Finally, we observe that 
$$ \partial_4 \Omega_1
=\partial_4 \Omega_2=\left\{ 
(1, 1, 0, 0, 1),(1,1,0,1,1)
\right\},$$
$$ \partial_4 \Omega_3
=\left\{ 
(0,0,1,0,1),(0,1,1,0,1),(1,0,1,0,1)
\right\},$$ and 
$$ \partial_4 \Omega_5
=\left\{ \begin{array}{l}
(0,0,1,0,1),(0,0,1,1,1),(0,1,1,0,1),(1,0,1,0,1),(1,1,0,0,1),\\(0,1,1,1,1),(1,0,1,1,1),(1,1,0,1,1),(1,1,1,0,1),(1,1,1,1,1) \end{array}
\right\}.$$ 
Similarly, we find that
$$ B_5
=\left\{ \begin{array}{l}
(0,0,1,0,1),(0,0,1,1,1),(0,1,1,0,1),(1,0,1,0,1),(1,1,0,0,1),\\ (0,0,1,0,0),(0,0,1,1,0),(0,1,1,0,0),(1,0,1,0,0),(1,1,0,0,0)\end{array}
\right\}$$ 
and 
$$I_5^{(p)}=p(1-p)^3+4p^2(1-p)^2+4p^3(1-p)+p^4.$$
Therefore, for this code, 
\begin{align*}
I^{(p)}&=
4 \left( p(1-p)^3+4p^2(1-p)^2+4p^3(1-p)+p^4 \right)\\
&\hphantom{=}+2p(1-p)^3+6p^2(1-p)^2+4p^3(1-p) \\
&=6p(1-p)^3+22p^2(1-p)^2+20p^3(1-p). 
\end{align*}
\end{example}

In Example \ref{toy_ex}, we see that $S_i=S_{i'}$ implies $\Omega_i=\Omega_{i'}$. While this fact follows immediately from the definitions, the example highlights its impact on $B_j$ when $\Omega_i=\Omega_j$. In particular, we see that $S_1=S_2,\, \Omega_1 \sqcup \Omega_1+e_1 =B_1$ and $B_2=\Omega_2 \sqcup \Omega_2+e_2 =B_2$; also, $S_4=S_5,\,$ $\Omega_4 \sqcup \Omega_4+e_4=B_4,\,$ and $\Omega_5 \sqcup \Omega_5+e_5 =B_5.$

\begin{remark} \label{wt_enum_rem}
We note that the influence of the $j^\mathrm{th}$ coordinate depends on the weights of words in $B_j$ rather than the words themselves. That might suggest that influences of a code with known weight enumerator $\sum_{c \in C} x^{\wt(c)}$ are easy to calculate. However, we are reminded that the weights needed are not necessarily those of codewords but instead of words in $B_j$. In Example \ref{toy_ex}, we see that $B_j \setminus C \neq \emptyset$ for all $j \in [5]$. 
\end{remark}

Despite the word of caution in Remark \ref{wt_enum_rem}, knowledge of the weights of codewords can provide some information about influences in certain circumstances. Suppose
$S_i=S_{j}$ for some $i \neq j$. Then $\Omega_i=\Omega_{j}$. Moreover,  we see that 
$$\left\{ x \in \Omega_i \colon x+e_j \notin \Omega_i \right\} = \Omega_i = \Omega_j$$
and 
$$\left\{ x \notin \Omega_i \colon x+e_j \in \Omega_i \right\} = \Omega_i-e_j = \Omega_j-e_j.$$
Consequently, recalling that $S_j \subseteq \Omega_j$, we can make the following observation. 

\begin{obs}
If $S_i=S_j$ for some $i \neq j$, $$S_j \sqcup S_j+e_j \subseteq
\Omega_j \sqcup \Omega_j+e_j \subseteq B_j.
$$ In this case, $$I_j^{(p)} \geq \sum_{x \in S_j} p^{\wt({x}) - 1}(1-p)^{n - \wt({x})}.$$
\end{obs}

This notion of influence of coordinate of a code is connected to that of monotone Boolean functions. Consider the characteristic function of $\Omega_i$:
$$
\begin{array}{lllc}
\chi_{\Omega_i}: &\F_2^n &\rightarrow & \F_2 \\ 
&x & \mapsto & \begin{cases} 1 & \textnormal{if } x \in \Omega_i \\ 0 & \textnormal{otherwise}. \end{cases}
\end{array}
$$
Notice that the \(j^\text{th}\) boundary of $\Omega_i$ is \[
\partial_j \Omega_i = \{{x} \in \FF_2^n \,:\, \chi_{\Omega_i}({x}) \neq \chi_{\Omega_i}({x} + {e}_j)\}
\]
Observe that if $\partial_j\Omega_i \neq \partial_j\Omega_\ell$ for all $i, \ell \in [n] \setminus \{ j \}$ with $i \neq \ell$, then 
$$I^{(p)}_j(C) = p^{-1} \sum_{i \in [n] \setminus \{ j \}} I^{(p)}_j(\chi_{\Omega_i}).$$

\section{Codes with minimum disjoint support} \label{min_section}

In this section, we consider a families of codes 
whose coordinates may be partitioned into sets which support so-called {\bf minimal codewords}, meaning those whose support does not contain the support of another nonzero codeword \cite{minimal}. We will see that this family of codes contains the familiar repetition codes as well as some distinct weight codes. Then, in the next section, we will demonstrate that this condition makes it straightforward to compute the influences of variables of these codes and determine the total influence. As a consequence, we provide expressions for influences of repetition codes and some relatives. 

We begin by recalling some codes from the literature. The {\bf $r$-times repetition code} $C_r$ may be defined by the generator matrix
$$G_{r}:= \left[ 
\begin{array}{cccc}
1_{r} &  & & \\
& 1_{r} & & \\
& & \ddots & \\
& & & 1_{r} 
\end{array}
\right] \in \F_2^{k \times rk}$$
expressed in block form where $1_r = 1 1 \cdots 1$ is a block of $r$ $1$s and all other entries in the matrix are zero. Note $C_r$ is an $[rk,k,r]$ code. The codewords of $C_r$ are of the form
$$x G_r = (\underbrace{x_1, \dots, x_1}_r, \dots, \underbrace{x_k, \dots, x_k}_r) \in \F_2^{rk}$$ where $x \in \F_2^k$ and each coordinate is repeated $r$ times. 

The next codes that will be considered were introduced in \cite{HH} as distinct weight codes (see also \cite{how_many}). 
They are reminiscent of repetition codes, but they lack the symmetry that $C_r$ features as coordinates are repeated different numbers of times depending on their position. A code $C$ of length $n$ is called a {\bf distinct weight code} if the map
$$
\begin{array}{cccc}
\wt: &C &\rightarrow& [n] \\
&c &\mapsto & \wt(c)
\end{array}
$$ 
is injective. Consider the block diagonal matrix given by 
$$G_{r,k}:= \left[ 
\begin{array}{cccc}
1_{2^r} &  & & \\
& 1_{2^{r+1}} & & \\
& & \ddots & \\
& & & 1_{2^{r+k-1}} 
\end{array}
\right]  \in \F_2^{k \times 2^r (2^{k}-1)}.$$
Note that the $i^\text{th}$ row of $G_{r,k}$ has weight $2^{r+i-1}$ for all $i \in [k]$. Let $C_{r,k}$ denote the binary code with generator matrix $G_{r,k}$. The codewords of $C_{r,k}$ are of the form
$$x G_{r,k} = (\underbrace{x_1, \dots, x_1}_{2^r}, \underbrace{x_2, \dots, x_2}_{2^{r+1}} \dots, \underbrace{x_k, \dots, x_k}_{2^{r+k-1}}) \in \F_2^{2^{r+k}-2^r}$$ where $x \in \F_2^k$. Hence, $C_{r,k}$ is a $[2^{r+k}-2^r,k,2^r]$ code. 

We are also interested in a hybrid between the repetition codes and distinct weight codes introduced above. We may observe that the repetition code and distinct weight code mentioned above may be seen as these hybrid codes. Even so, it can be convenient to consider their more refined forms to reflect the impact of the more controlled structures. Consider a partition of $[n]$ into parts $A_1, \dots, A_k$:
$$
[n] = A_1 \sqcup A_2 \sqcup \dots \sqcup A_k.
$$
Let $C_A$ denote the {\bf hybrid code} with generator matrix 
$$\left[ 
\begin{array}{cccc}
1_{A_1} &  & & \\
& 1_{A_2} & & \\
& & \ddots & \\
& & & 1_{A_k} 
\end{array}
\right] \in \F_2^{k \times \sum_{i=1}^k | A_i |}.$$
Then $C_A$ is an $[n,k,\min \{ |A_i| \colon i \in k \}]$ code.

Next, we introduce the family of codes we wish to study. 

\begin{definition} A binary $[n,k,d]$ code $C$ is said to have {\bf minimum disjoint support} provided:
\begin{enumerate}
\item  For each $i \in [n]$, 
the set of codewords $S_i$
                    has a minimum element $u_i$ according to $\prec$.
\item There exist   $i_1, \dots, i_s \in [n]$ such that 
 $$[n] = \supp(u_{i_1}) \sqcup \dots \sqcup \supp(u_{i_s}).$$
\end{enumerate}
The codewords $u_i$ are called {\bf minimum support codewords}.
\end{definition}

We will demonstrate that each of the codes introduced at the beginning of this section is a minimum disjoint support code. We will identify minimum elements $u_i$ of the $S_i$, $i \in [n]$, and a partition 
 $$[n] =T_{i_1} \sqcup \dots \sqcup T_{i_s}$$
 where $T_i:=\supp(u_i)$. We will see in Section \ref{inf_section} that these partitions will be very useful in finding influences of coordinates of codewords.

\begin{proposition}
    The $r$-times repetition code $C_r$, the distinct weight code $C_{r,k}$, and the hybrid code $C_A$ are minimum disjoint support codes.
\end{proposition}

\begin{proof}
For the repetition code is $C_r$, we can take
$$u_{sr+1}=\dots=u_{(s+1)r}=\sum_{t=1}^r e_{sr+t}$$ and 
$$T_{sr+1}=\dots=T_{(s+1)r}= \left\{ sr+1, \dots, (s+1)r \right\}$$
for $s \in \left\{ 0, \dots, k-1 \right\}$. Hence, $$[n]= T_r \sqcup T_{2r} \sqcup \dots \sqcup T_{kr}.$$

Next, observe that the distinct weight code $C_{r,k}$ has $S_1=\dots=S_{2^r}$ and $S_{2^{r+s-1}+1}=\dots=S_{2^{r+s}}$ for all $1 \leq s \leq k-1$. 
Hence, $C_{r,k}$ is a minimum disjoint support code with minimal support codewords
$$
u_1=\dots=u_{2^r} = \sum_{\ell=1}^{2^r} e_\ell$$
and 
$$
u_{2^{r+s-1}+1}=\dots=u_{2^{r+s}} = \sum_{\ell=2^{r+s-1}+1}^{2^{r+s}} e_\ell$$ for $s \in \{ 1, \dots, k-1\}$. 
Hence, $$[n]= T_{2^r} \sqcup T_{2^{r+1}} \sqcup \dots \sqcup T_{2^{r+k-1}}.$$

Finally, the hybrid code $C_A$ has $$u_i=\sum_{j \in A_i} e_j$$ for each $i \in [n]$ and 
\[[n] = \bigsqcup_{i=1}^k T_{\ell_i}\] where we define
$\ell_i = \sum\limits_{j=1}^k | A_j |$.
\qedhere

\end{proof}

\begin{proposition}
\label{prop:DisjointSupport}
Let \(C\) be a binary \([n,k,d]\) minimum disjoint support code with minimum support codewords \(u_1, u_2, \hdots, u_n\). Then for any \(i,j \leq n\) either \(u_i = u_j\), or \(u_i\) and \(u_j\) have disjoint support.
\end{proposition}

\begin{proof}
Suppose that \(u_i\) and \(u_j\) do not have disjoint support, and let \(t \in \supp(u_i) \cap \supp(u_j)\). Since \(u_i\) is supported at \(t\), we must have \(u_t \preceq u_i\). Now, if \(i \not\in \supp(u_t)\), then \(u_i + u_t\) is supported at \(i\) and not at \(t\), contradicting our choice of \(t\) and our characterization of \(u_i\). Thus \(i \in \supp(u_t)\), and so by the minimality of \(u_i\), we must have \(u_i \preceq u_t\). Of course, this establishes that \(\supp(u_i) = \supp(u_t)\), and so \(u_i = u_t\), since \(C\) is binary. A precisely analogous argument establishes that \(u_j = u_t\), and consequently \(u_i = u_j\), as required.
\end{proof}

\begin{proposition}
    Given binary minimum disjoint support codes $C_1$ and $C_2$ of lengths $n_1$ and $n_2$, their product $C_1 \otimes C_2$ is also a minimum disjoint support code of length $n_1n_2$. 
\end{proposition}

\begin{proof}
    Let \(\{u_1, \hdots, u_{n_1}\}\) and \(\{v_1, \hdots, v_{n_2}\}\) be the minimum suppport codewords of \(C_1\) and \(C_2\), respectively. We claim that \(\{u_i\otimes v_j\}_{\substack{1 \leq i \leq {n_1} \\ 1 \leq j \leq n_2}}\) are minimum support codewords for \(C_1 \otimes C_2\).

    First observe that each \(u_i \otimes v_j \in C_1 \otimes C_2\). Next, choose \(r \in [n_1]\) and \(t\in [n_2]\), and let \(c \in C_1 \otimes C_2\) be such that \(c_{(r,t)} = 1\). We may write
    \begin{align}
    \label{eqn:tensordecomposition}
    c = \sum_{i=1}^N c_i^1 \otimes c_i^2
    \end{align}
    for some \(N \in \NN\) and some \(c_1^1, c_2^1, \hdots, c_N^1 \in C_1\) and \(c_1^2, c_2^2, \hdots, c_N^2 \in C_2\). Let \(n\) be the number of terms \( c_i^1 \otimes c_i^2\) with \( (c_i^1 \otimes c_i^2)_{(r,t)} = 1\), and let \(m = N - n\). Let \(1 \leq i_1 < i_2 < \cdots < i_n \leq N\) be the indices which satisfy \( (c_{i_j}^1 \otimes c_{i_j}^2)_{(r,t)} = 1\), and let \(1 \leq i'_1, < i'_2 < \cdots < i'_{m} \leq N\) be the remaining indices. Defining
    \begin{align*}
    x_j^1 &= c_{i_j}^1  \mbox{ for } j = 1,2,\hdots, n\\
    x_j^2 &= c_{i_j}^2  \mbox{ for } j = 1,2,\hdots, n\\
    y_j^1 &= c_{i'_j}^1 \mbox{ for } j = 1,2,\hdots, m\\
    y_j^2 &= c_{i'_j}^2 \mbox{ for } j = 1,2,\hdots, m
    \end{align*}
    we obtain the following expression for \(c\):
    \[
        c = \sum_{i=1}^n x_i^1 \otimes x_i^2 + \sum_{i=1}^m y_i^1 \otimes y_i^2 
    \]
    where each \(x_i^1, y_i^1 \in C_1\), \(x_i^2, y_i^2 \in C_2\), each \((x_i^1 \otimes x_i^2)_{(r,t)} = 1\) and each \((y_i^1 \otimes y_i^2)_{(r,t)} = 0\). Obviously \(n\) is odd. As well, we have
    \begin{align*}
        (x_i^1 \otimes x_i^2)_{(r,t)} = 1 \implies x_{i,r}^1 = 1 \mbox{ and } x_{i,t}^2 = 1
    \end{align*}
    so that \(u_r \preceq x_i^1\) and \(v_t \preceq x_i^2\) for all \(i\). Consequently, \( u_r \otimes v_t \preceq x_i^1 \otimes x_i^2\) for all \(i\), and so (since \(n\) is odd), 
    \[
    u_r \otimes v_t \preceq \sum_{i=1}^n x_i^1\otimes x_i^2.
    \]

    It remains to show that \(\left(\sum_{i=1}^m y_i^1 \otimes y_i^2\right)_{(r',t')} = 0\) for all \( (r',t') \in \supp(u_r \otimes v_t)\), since this will yield \(u_r \otimes v_t \preceq c\).

    Fix \(j\), and suppose that \( (y_j^1 \otimes y_j^2)_{(r',t')} = 1\) for some \( (r',t') \in \supp(u_r \otimes v_t)\). Obviously this requires \(y_{j,r'}^1 = 1\) and \(y_{j,t'}^2 = 1\), so that \(u_{r'} \preceq y_j^1\) and \(v_{t'} \preceq y_j^2\). But \( (s',t') \in \supp(u_r \otimes v_t)\) says that \(r' \in \supp(u_s)\) and \(t' \in \supp(v_t)\), so that by Proposition~\ref{prop:DisjointSupport}, we in fact have \(u_r \preceq y_j^1\) and \(v_t \preceq y_j^2\). But this gives \( (y_j^1 \otimes y_j^2)_{(r,t)} = 1\), contradicting the assumption that \((y_j^1 \otimes y_j^2)_{(r,t)} = 0\). So we must have \( (y_j^1 \otimes y_j^2)_{(r',t')} = 0\) for all \( (r',t') \in \supp(u_r \otimes v_t)\), and so \(\left(\sum_{i=1}^m y_i^1 \otimes y_i^2\right)_{(r',t')} = 0\) for all \( (r',t') \in \supp(u_r \otimes v_t)\), as required.

    It follows that \( (u_r \otimes v_t) \preceq c\), as required.

    Finally, for any \(r \in [s_1]\) and \(t \in [s_2]\), we have \(\supp(u_r \otimes v_t) = T_i^1 \times T_j^2\). Moreover, if \([n_1] = \bigsqcup_{i \in \{i_1, \hdots, i_a\}} T_i^1\) and \([n_2] = \bigsqcup_{j \in \{j_1, \hdots, j_b\}} T_j^1\), then  
\[
[n_1] \times [n_2] = \bigsqcup_{\substack{i \in \{ i_1, \dots, i_{a} \} \\[2pt] j \in \{ j_1, \dots, j_{b} \}}} T^1_i \times T^2_j = \bigsqcup_{\substack{i \in \{ i_1, \dots, i_{a} \} \\[2pt] j \in \{ i_1, \dots, i_{b} \}}} \supp(u_i \otimes v_j)
\]
which completes the proof.
\end{proof}

\begin{remark}
It is worth distinguishing minimum disjoint support codes from minimal codes, which are codes in which every codeword covers only itself \cite{minimal}. More formally, a code $C$ is minimal provided $c \preceq c'$ for codewords $c, c' \in C \setminus \{ 0 \}$ implies $c=c'$. Notice that the $2$-times repetition code of length $6$ is not minimal since $(1,1,0,0,0,0), (1,1,1, 1,0,0) \in C$ and $(1,1,0,0,0,0) \preceq (1,1,1, 1,0,0)$.
\end{remark}

We will see that the structure of minimum disjoint support codes lends itself to determining the sets $\Omega_i$ and their boundaries, facilitating the calculation of influences.

\section{Influences of coordinates of parity-check and minimum disjoint support codes} \label{inf_section}

In this section, we compute influences of some families of codes. We begin with a result on influences of the coordinates of a simple parity-check code and then consider minimum disjoint support codes.

Consider the {\bf simple parity-check code $C$} of length $n$ which is given by parity-check matrix $$H = [1 \cdots 1] \in \F_2^{1 \times n},$$ so that
$$
C = \left\{ c \in \F_2^n\colon \sum_{i=1}^n c_i = 0 \right\}.$$
Then $C$ is an $[n, n-1,2]$ code as codewords are of the form $(c_1, \dots, c_{n-1}, \sum_{i=1}^{n-1} c_i)$. 

\begin{proposition} \label{pc_inf_prop}
Given an $[n, n-1,2]$ simple parity-check code, for all $j \in [n]$, the influence of the $j^{th}$ coordinate is 
$$I_j^{(p)} = (n-1)(1-p)^{n-2}.$$
Therefore, 
$$I^{(p)}=n(n-1)(1-p)^{n-2}.$$
\end{proposition}

\begin{proof}
For $i \in [n]$, $S_i = \left\{ x \in \F_2^n\colon x_i = 1, \wt(x) \mbox{ is even} \right\}$. Then 
$$
\begin{array}{lll}
\Omega_i &=& \left\{ x \in \F_2^n\colon x_i = 1, \wt(x) \geq 2 \right\} \\
& = & \left\{ x \in \F_2^n\colon x_i = 1 \right\} \setminus \left\{ e_i \right \}.
\end{array}
$$ 
Notice that for $i \in [n] \setminus \{ j \}$, $e_{ij}:=e_i+e_j \in \Omega_i$ but $e_{ij}+e_j = e_i \notin \Omega_i$ since $\wt(e_i)=1$. It follows that 
 $$\partial_j \Omega_i = \{ e_{ij}, e_i \}$$ and 
$$B_j = \bigcup_{i \in [n] \setminus \{ j \} }  \{  e_i, e_{ij} \}.$$
Since $B_j$ contains $n-1$ words $e_{ij}$ of weight $2$ and $n-1$ words $e_i$ of weight $1$, 
$$I_j^{(p)} = (n-1) \left( p(1-p)^{n - 2} + (1-p)^{n-1} \right) = (n-1)(1-p)^{n-2}.$$
Therefore, 
\[I^{(p)}=n(n-1)(1-p)^{n-2}. \qedhere\]
\end{proof}

Next, we consider a particular example to illustrate the observations made in the proof of Proposition \ref{pc_inf_prop}. 

\begin{example} \label{pc_ex}
    Consider the $[10,9,2]$ parity-check code $C$ and the influence of first coordinate on the other nine coordinates. Take $j=1$ and first consider $i=2$. Observe that $$S_2=\left\{ c \in C\colon c_2 \neq 0\right\}=\left\{ \left(c_1, 1, c_3, \dots, c_{10} \right)\colon c_1+ c_3+ \dots+ c_{10} = 1 \right\} %\textnormal{ is odd} \right\}
    $$ 
    and $$\Omega_2=\left\{ \left(x_1, 1, x_3, \dots, x_{10} \right)\colon  \left(x_1, x_3, \dots, x_{10} \right) \neq (0, \dots, 0) \right\}.$$
    Then $$\left\{ x \in \Omega_2\colon x+e_1 \notin \Omega_2 \right\}=\left\{ (1,1,0,\dots, 0) \right\}$$ and $$\left\{ x \notin \Omega_2\colon x+e_1 \in \Omega_2 \right\}=\left\{ (0,1,0,\dots, 0) \right\}.$$ Hence, $$\partial_1 \Omega_2=
    \left\{  (1,1,0,\dots, 0), (0,1,0,\dots, 0) \right\}.$$ Similarly, we find that 
    $$
    \begin{array}{ccl}
    \partial_1 \Omega_3&=&
    \left\{  (1,0,1, 0,\dots, 0), (0,0, 1,0,\dots, 0) \right\}\\
    \partial_1 \Omega_4&=&
    \left\{  (1,0,0, 1, 0,\dots, 0), (0,0, 0,1,0,\dots, 0) \right\}\\
    & \vdots & \\
    \partial_1 \Omega_{10}&=&
    \left\{  (1, 0,\dots, 0, 1), (0,\dots, 0, 1) \right\}.
    \end{array}
    $$
    Hence, $$B_1 = \left\{ e_1+e_i\colon i \in \{ 2, \dots, 10 \} \right\} \cup \left\{ e_i \colon i \in \{ 2, \dots, 10 \right\}\}.$$
   As a result, $$I_1^{(p)} = 
   9p^{2-1}(1-p)^{10-2}+9p^{1-1}(1-p)^{10-1}=9p(1-p)^8+9(1-p)^9=9(1-p)^8.$$ The influences of other coordinates may be found similarly. Alternatively, one may observe that $I_j^{(p)}=I_{j'}^{(p)}$ for all $j, j' \in [10]$. Therefore, 
   $$I^{(p)}=10\cdot 9(1-p)^8=90(1-p)^8.$$
\end{example}

According to Proposition~\ref{pc_inf_prop}, all coordinates have the same influence. This fact also follows from the double transitivity of the automophism group of the code (see, for instance, \cite{Kumar_16}). 

Next, we consider influences of coordinates of the codes introduced in Section \ref{min_section}. 

\begin{theorem} \label{min_thm}
Consider a binary $[n,k,d]$ code with minimum disjoint support given by $[n] = T_{i_1} \sqcup \dots \sqcup T_{i_s}$. Then the influence of coordinate $j$ where $j \in T_i$ is
$$
I_j^{(p)} = p^{| T_i |-2}.
$$
The total influence  is given by
$$I^{(p)}=\sum_{\ell=1}^s | T_{i_\ell} | p^{| T_{i_\ell} | -2}.$$
\end{theorem}

\begin{proof}
For all $i \in [n]$, $$\Omega_i = \left\{ \sum_{\ell \in T_i} e_\ell +  \sum_{\ell \in [n] \setminus T_i} a_\ell e_\ell\colon a_\ell \in \F_2 \right\} \cong \F_2^{n-| T_i |}.$$
Hence, 
$$\partial_j \Omega_i = \begin{cases}
\Omega_j \sqcup \left( \Omega_j - e_j \right) & \textnormal{if } j \in T_i \\ \emptyset & \textnormal{otherwise.} \end{cases}$$ It follows that
$$B_j = \bigcup_{\substack{i \in [n] \\[2pt] j \in T_i}} \left( \Omega_j \sqcup \left( \Omega_j - e_j \right) \right) = \Omega_j \sqcup \left( \Omega_j - e_j \right).$$
Then 
\begin{align*}
B_j &=  \left\{ \sum_{\ell \in T_i} e_\ell +  \sum_{\ell \in [n] \setminus T_i} a_\ell e_\ell\colon a_\ell \in \F_2 \right\}  \sqcup  \left\{ \sum_{\ell \in T_i \setminus \{j \}} e_\ell +  \sum_{\ell \in [n] \setminus T_i} a_\ell e_\ell\colon a_\ell \in \F_2 \right\}\\ \ \\
&=  \left\{ \sum_{\ell \in T_i \setminus \{ j \}} e_\ell +  \sum_{\ell \in \{j \} \cup [n] \setminus T_i} a_\ell e_\ell\colon a_\ell \in \F_2 \right\}.
\end{align*}
Notice that $B_j \cong \F_2^{n-| T_i | +1}$. Thus, if $| T_i | \geq 2$ for all $i \in [n]$ and $j \in T_i$,
\begin{align*}
I_j^{(p)} &= \sum_{\ell=0}^{n-| T_i |+1} \binom{{n-| T_i |+1}}{\ell}p^{| T_i |+\ell-2}(1-p)^{n-| T_i |-\ell+1}  \\ 
& =  p^{| T_i | -2} \sum_{\ell=0}^{n-| T_i |+1} \binom{{n-| T_i |+1}}{\ell}p^{\ell}(1-p)^{n-| T_i |-\ell+1}  \\ 
& =  p^{| T_i | -2} (p+(1-p))^{n-| T_i |+1}= p^{| T_i |-2}.
\end{align*}
Hence, 
\[
I^{(p)}=\sum_{l=1}^s | T_{i_l} | p^{| T_{i_l} | -2}. \qedhere
\]
\end{proof}

\begin{corollary} \label{rep_cor}
An $r$-times repetition code of length $n$ with $r \geq 2$ has 
$$I_j^{(p)}=p^{r-2}$$ for all $j$. Hence, 
$$I^{(p)}=n p^{r-2}.$$
\end{corollary} 

\begin{proof}
Note that $| T_i | = r$ for all $i \in [n]$. 
According to Theorem~\ref{min_thm}, $$I_j^{(p)}=p^{r-2}$$ for all $j$ and
\[I^{(p)}=n p^{r-2}. \qedhere\]
\end{proof}

\begin{example} \label{rep_ex}
    Consider the $[15,5,3]$ $3$-times repetition code $C_3$. 
    We set out to determine the influence of the first coordinate on the others, meaning we take $j=1$.
    Notice that 
    $$
    \begin{array}{ccccccl}
    S_1&=&S_2&=&S_3&=&C \cap \left\{ (1, 1, 1, x_4, \dots, x_{15})\colon x_i \in \F_2\right\}\\
    S_4&=&S_5&=&S_6 & = & C \cap \left\{ (x_1, x_2, x_3, 1, 1, 1, x_7, \dots, x_{15})\colon x_i \in \F_2\right\}\\
S_7&=&S_8&=&S_9 & = & C \cap \left\{(x_1, \dots, x_6, 1, 1, 1, x_{10}, \dots, x_{15})\colon x_i \in \F_2\right\}\\
    S_{10}&=&S_{11}&=&S_{12} & = & C \cap \left\{(x_1, \dots, x_9, 1, 1, 1, x_{13}, x_{14}, x_{15})\colon x_i \in \F_2\right\}\\
     S_{13}&=&S_{14}&=&S_{15} & = & C \cap \left\{(x_1, \dots,  x_{12}, 1,1,1)\colon x_i \in \F_2\right\}.
     \end{array}
    $$
    Then  $$
    \begin{array}{ccccccl}
\Omega_1&=&\Omega_2&=&\Omega_3&=& \left\{ (1, 1, 1, x_4, \dots, x_{15})\colon x_i \in \F_2\right\}\\
\Omega_4&=&\Omega_5&=&\Omega_6 & = & \left\{ (x_1, x_2, x_3, 1, 1, 1, x_7, \dots, x_{15})\colon x_i \in \F_2\right\}\\
\Omega_7&=&\Omega_8&=&\Omega_9 & = &  \left\{(x_1, \dots, x_6, 1, 1, 1, x_{10}, \dots, x_{15})\colon x_i \in \F_2\right\}\\
 \Omega_{10}&=&\Omega_{11}&=&\Omega_{12} & = &  \left\{(x_1, \dots, x_9, 1, 1, 1, x_{13}, x_{14}, x_{15})\colon x_i \in \F_2\right\}\\
\Omega_{13}&=&\Omega_{14}&=&\Omega_{15} & = &  \left\{(x_1, \dots,  x_{12}, 1,1,1)\colon x_i \in \F_2\right\}.
     \end{array}
    $$
For $i \in \{ 2, 3\}$, 
$$\left\{ x \in \Omega_i\colon x+e_1 \notin \Omega_i \right\}=\Omega_1$$
and $$\left\{ x \notin \Omega_i\colon x+e_1 \in \Omega_i \right\}= \left\{ (0, 1, 1, x_4, \dots, x_{15})\colon x_i \in \F_2\right\}.$$For $i \in \{ 4, 5, 6 \}$, 
$$\left\{ x \in \Omega_i\colon x+e_1 \notin \Omega_i \right\}=\left\{ x \notin \Omega_i\colon x+e_1 \in \Omega_i \right\}=\emptyset. $$
Thus, for $i \in \{2, 3\}$
$$\partial_1 \Omega_i =
\Omega_i \cup \left\{ (0, 1, 1, x_4, \dots, x_{15})\colon x_i \in \F_2\right\}
$$
and for $i \in [15] \setminus \{ 1, 2, 3\}$
$$\partial_1 \Omega_i =\emptyset.
$$
This implies 
$$
B_1= \Omega_1 \sqcup \left\{ (0, 1, 1, x_4, \dots, x_{15})\colon x_i \in \F_2\right\}
=\left\{ (x_1, 1, 1, x_4, \dots, x_{15})\colon x_i \in \F_2\right\}.$$
Therefore, 
\begin{align*}
I_1^{(p)}&=\sum_{x \in \F_2^{13}} p^{\wt(x)+2-1}(1-p)^{15-\wt(x)-2}=\sum_{\ell =0}^{13} \binom{13}{\ell}p^{\ell + 1}(1-p)^{13-\ell}\\ 
&=p
\sum_{\ell =0}^{13} \binom{13}{\ell}p^{\ell }(1-p)^{13-\ell}=p(p+(1-p))^{13}=p.
\end{align*}
Similarly, $$I_j^{(p)}=p$$ for all $j\in [n]$. Hence, all coordinates have the same influence. 
Moreover, $$I^{(p)}=15p.$$

Alternatively, to find these influences, one may apply directly Theorem \ref{min_thm} noting that $| T_{i_\ell} | =3$ for all $i_\ell$, $\ell \in [5]$.

We note that $C_3$ is not a distinct weight code, since $(1,1,1,0, \dots, 0),(0, \dots, 0, 1, 1, 1) \in C_3$ and $\wt (1,1,1,0, \dots, 0)= \wt(0, \dots, 0, 1, 1, 1)$.

\end{example}

According to Corollary~\ref{rep_cor} and as illustrated in Example \ref{rep_ex}, all of the coordinates of an $r$-times repetition code have the same influence. Next, we see that this is not necessarily the case for the distinct weight codes. 

\begin{corollary} \label{dwc_cor}
For the distinct weight code $C_{k,r}$, the influence of coordinate $j \in [2^r (2^{k}-1)]$ is
$$I_j^{(p)} = p^{2^{ \lceil \log_2j \rceil}-2}.
$$ 
The total influence is
$$I^{(p)}= \sum_{i=1}^k | T_{2^{ir}} | p^{| T_{2^{ir}} |-2} = \sum_{i=1}^k 2^{ir} p^{2^{ir}}.$$
\end{corollary}

\begin{proof}
Notice that $$| T_{2^{r+s}} | = 2^{r+s}$$ for $s\in \{ 0, \dots, k-1\}$. In fact, $j \in T_{2^{ \lceil \log_2j \rceil}}$. 
According to Theorem~\ref{min_thm}, $$
I_j^{(p)} = p^{2^{ \lceil \log_2j \rceil}-2}
$$ 
and 
\[I^{(p)}= \sum_{i=1}^k | T_{2^{ir}} | p^{| T_{2^{ir}} |-2} = \sum_{i=1}^k 2^{ir} p^{2^{ir}}. \qedhere\]
\end{proof}

\begin{example}
Consider the  distinct weight code $C_{3,2}$ which has generator matrix $$
G_{3,2}=\left[
\begin{array}{ccc}
1_4 & 0_{8} & 0_{16}\\
0_4 & 1_{8} & 0_{16}\\
0_4 & 0_{8} & 1_{16}
\end{array}
\right]. 
$$ Recall that $C_{3,2}$ is a $[28,3,4]$ code. Indeed, 
$$C_{3,2}= \left\{
0_{28}, 1_4 0_{24}, 0_4 1_{8} 0_{16}, 1_{12} 0_{16}, 0_{12} 1_{16}, 1_4 0_8 1_{16}, 0_4 1_{24},  1_{28} \right\}, 
$$
demonstrating that it is fact a distinct weight code with exactly one word of weight \(0, 4, 8,\allowbreak 12, 16, 20, 24, 28.\) Moreover, 
$$
\begin{array}{lclclclcl}
S_1&=&S_2&=&S_3&=&S_4&=&\left\{ 1_4 0_{24}, 1_{12} 0_{16},  1_4 0_8 1_{16},   1_{28} \right\}\\
S_5&=&S_6&=&\cdots&=&S_{12}&=&\left\{ 0_4 1_{8} 0_{16}, 1_{12} 0_{16},  0_4 1_{24},  1_{28}\right\}\\
S_{13}&=&S_{14}&=&\cdots&=&S_{28}&=&\left\{ 0_{12} 1_{16}, 1_4 0_8 1_{16}, 0_4 1_{24},  1_{28} \right\}
 \end{array}
$$
and
$$
\begin{array}{lclclclcl}
\Omega_1&=&\Omega_2&=&\Omega_3&=&\Omega_4&=& \left\{ 
(1, 1, 1, 1, x_5, \dots, x_{28})\colon x_i \in \F_2 \right\},\\
\Omega_5&=&\Omega_6&=&\cdots&=&\Omega_{12}&=& \left\{ 
(x_1, x_2, x_3, x_4, 1, \dots, 1, x_{13}, \dots, x_{28})\colon x_i \in \F_2 
\right\},\\
\Omega_{13}&=&\Omega_{14}&=&\cdots&=&\Omega_{28}&=&\left\{
 (x_1, \dots, x_{12}, 1, \dots 1)\colon x_i \in \F_2  \right\} \subseteq \F_2^{28}.
 \end{array}
$$
Take first $j=1$. For $i \in \{ 2, 3, 4 \}$, 
$$\left\{ x \in \Omega_i\colon x+e_1 \notin \Omega_i \right\}=\Omega_1$$
and $$\left\{ x \notin \Omega_i\colon x+e_1 \in \Omega_i \right\}=\Omega_1 - e_1 = \left\{ (0, 1, 1, 1, x_5, \dots, x_{28})\colon x_i \in \F_2\right\}.$$For $i \in \{ 5, \dots, 28 \}$, 
$$\left\{ x \in \Omega_i\colon x+e_1 \notin \Omega_i \right\}=\left\{ x \notin \Omega_i\colon x+e_1 \in \Omega_i \right\}=\emptyset. $$
Thus, 
$$\partial_1 \Omega_i =
\begin{cases}
\left\{ (0, 1, 1, 1, x_5, \dots, x_{28})\colon x_i \in \F_2\right\} & \textnormal{for } i = 2, 3, 4 \\
\emptyset & \textnormal{otherwise.}
\end{cases}
$$
This implies 
$$
B_1= \cup_{i \in \{ 2, \dots, 28 \}} \partial_1 \Omega_i = \Omega_1 \sqcup \Omega_1-e_1
=\left\{ (x_1, 1, 1, 1, x_5, \dots, x_{28})\colon x_i \in \F_2\right\}.$$
Therefore, 
\begin{align*}
I_1^{(p)}&=\sum_{x \in \F_2^{25}} p^{\wt(x)+3-1}(1-p)^{28-\wt(x)}=\sum_{\ell =0}^{25} \binom{25}{\ell}p^{\ell + 2}(1-p)^{25-\ell}\\ 
&=p^2 
\sum_{\ell =0}^{25} \binom{25}{\ell}p^{\ell }(1-p)^{25-\ell}=p^2(p+(1-p))^{25}=p^2.
\end{align*}
Moreover, $$I_2^{(p)}=I_3^{(p)}=I_4^{(p)}=p^2.$$Using similar observations, one may calculate the remaining $I_j^{(p)}$. Since $| \{ 5, \dots, 12 \} | = 8$ and $| \{ 13, \dots, 28 \} | =16$, we obtain 
$$I_5^{(p)}=\dots=I_{12}^{(p)}=p^{8-2}=p^6$$
and 
$$I_{13}^{(p)}=\dots=I_{28}^{(p)}=p^{16-2}=p^{14},$$ indicating that coordinates with 
larger influences are those with smaller indices. Moreover, 
$$I^{(p)}=4p^2+8p^6+16p^{14}.$$
These conclusions are consistent with those found by apply directly Theorem \ref{min_thm} noting that 
$| T_{i_1} | =4$, 
$| T_{i_2} | =8$, and $| T_{i_3} | =16$.
\end{example}

\begin{corollary} \label{hybrid_inf_cor}
Consider the hybrid code $C_A$ of length $n$ with $| A_i | \geq 2$ for all $i \in [k]$. Then for $j \in A_i$, $$I_j^{(p)}=p^{|A_i|-2}.$$ 
Moreover, 
$$I^{(p)}=\sum_{i=1}^k | A_i | p^{| A_i | -2}.$$
\end{corollary}

\begin{proof}
The result follows similarly to the proof of Corollary~\ref{dwc_cor}.
\end{proof}

We note that coordinates with larger influences have indices corresponding to smaller parts of the partition $[n]=A_1 \sqcup \dots \sqcup A_k$.

As observed previously, influences of coordinates in the simple parity-check codes and repetition codes were identical. In contrast, the influences of coordinates in the distinct weight codes $C_{k,r}$ may differ. In the next example, we consider hybrid codes to demonstrate how much influences can differ for these code families.  

\begin{example}
    Consider an integer $n \geq 3$ and the partition 
    $$[n]=[2] \sqcup \{ 3, \dots, n \}.$$
    According to Corollary \ref{hybrid_inf_cor}, 
    $$I_1^{(p)}=I_2^{(p)}=1$$ whereas
    $$I_j^{(p)}=p^{n-3} \ \forall j \in [n] \setminus \{ 1, 2 \}.$$
\end{example}

\section{Conclusion} \label{conclusion_section}

In this paper, we defined minimum disjoint support codes, noting that repetition codes and some distinct weight codes have minimum disjoint support. We reviewed the concept of influences of variables of monotone Boolean functions and explained its connection to coding theory. Finally, we determined the influences of coordinates of some families of error correcting codes, including simple parity check codes and codes with minimum disjoint support. While the codes themselves have rates approaching $0$ as the length goes to infinity, we hope that this study provides insight into influences of other code families with more promising rate properties. 

\bibliography{bib}{}
\bibliographystyle{abbrv}

\end{document}